\documentclass[11pt]{article}

\usepackage{amsmath, amsthm, amssymb, graphicx, enumerate, fullpage, url}
\usepackage{color, hyperref}

\hypersetup{
	colorlinks=true,
	linkcolor=blue,
	citecolor=magenta
}

\theoremstyle{plain}
\newtheorem{theorem}{Theorem}[section]
\newtheorem{proposition}[theorem]{Proposition}
\newtheorem{lemma}[theorem]{Lemma}
\newtheorem{corollary}[theorem]{Corollary}
\newtheorem{claim}[theorem]{Claim}

\theoremstyle{definition}
\newtheorem{definition}[theorem]{Definition}
\newtheorem{remark}[theorem]{Remark}

\newcommand{\E}{\mathbb{E}}
\newcommand{\F}{\mathbb{F}}

\newcommand{\cC}{\mathcal{C}}
\newcommand{\cL}{\mathcal{L}}

\newcommand{\modstar}[1]{~(\mathop{\rm{mod}^*} #1)}

\DeclareMathOperator{\poly}{poly}
\DeclareMathOperator{\Lift}{Lift}
\DeclareMathOperator{\LiftedRS}{LiftedRS}
\DeclareMathOperator{\RS}{RS}
\DeclareMathOperator{\Tr}{Tr}
\DeclareMathOperator{\Deg}{Deg}
\DeclareMathOperator{\supp}{supp}
\DeclareMathOperator{\Corr}{Corr}

\begin{document}

\title{List-decoding algorithms for lifted codes}
\author{Alan Guo
		\thanks{CSAIL, Massachusetts Institute of
Technology, 32 Vassar Street, Cambridge, MA, USA. {\tt aguo@mit.edu}. Research
supported in part by NSF grants CCF-0829672, CCF-1065125,
and CCF-6922462, and an NSF Graduate Research Fellowship}
		\and
		Swastik Kopparty
		\thanks{Department of Mathematics \& Department of Computer Science, Rutgers University. {\tt swastik.kopparty@rutgers.edu}. Research supported in part by a Sloan Fellowship and NSF CCF-1253886.}
		}

\date{}		

\maketitle

\begin{abstract}
Lifted Reed-Solomon codes are a natural affine-invariant family of
error-correcting codes which generalize Reed-Muller codes. They were known to 
have efficient local-testing and local-decoding algorithms (comparable to the 
known algorithms for Reed-Muller codes), but with significantly better rate.
We give efficient algorithms for list-decoding and local list-decoding of 
lifted codes.
Our algorithms are based on a new technical lemma, which says that codewords of 
lifted codes are low degree polynomials when viewed as univariate polynomials 
over a big field (even though they may be very high degree when viewed as 
multivariate polynomials over a small field).

\end{abstract}

\section{Introduction}
\label{section:introduction}

By virtue of their many powerful applications
in complexity theory, there has been much interest in the study of
error-correcting codes which support ``local"
operations. The operations of interest include
local decoding, local testing, local correcting, and
local list-decoding. Error correcting codes equipped
with such local algorithms have been useful, for example,
in proof-checking, private information retrieval,
and hardness amplification.

The canonical example of a code which supports all the above
local operations is the Reed-Muller code, which is a code based
on evaluations of low-degree polynomials. Reed-Muller codes
have nontrivial local algorithms across a wide range of parameters.
In this paper, we will be interested in the constant rate regime.
For a long time, Reed-Muller codes were the only known codes in this regime
supporting nontrivial locality. Concretely, for every constant integer $m$ and every
constant $R < \frac{1}{m!}$, there are Reed-Muller codes of arbitrarily large
length $n$, rate $R$, constant relative distance $\delta$, which are locally
decodable/testable/correctable from $(\frac12 - \epsilon) \cdot \delta$ fraction
fraction errors using $O(n^{1/m})$ queries. In particular, no nontrivial locality
was known for Reed-Muller codes (or any other codes, until recently) with rate
$R > 1/2$.

In the last few years, new families of codes were found which had interesting 
local algorithms in the high rate regime (i.e., with rate $R$ near $1$).
These codes include multiplicity codes~\cite{KSY, K}, lifted codes~\cite{GKS, Guo},
expander codes~\cite{HOW} and tensor codes~\cite{Viderman10}. Of these, lifted codes are the only ones
that are known to be both locally decodable and locally testable. This paper gives new and improved decoding and testing algorithms for lifted codes.

\subsection{Lifted Codes and our Main Result}

Lifted codes are a natural family of algebraic, affine-invariant codes
which generalize Reed-Muller codes. We give a brief introduction to these codes now\footnote{Technically
we are talking about lifted Reed-Solomon codes, but for brevity we refer to them as lifted codes.}.
Let $q$ be prime power, let $d < q$ and let $m > 1$ be an integer.
Define alphabet $\Sigma = \F_q$. 
We define the lifted code $\mathcal C = \mathcal C(q,d,m)$  to be a subset of 
$\Sigma^{\F_q^m}$, the space of functions from $\F_q^m$ to $\Sigma = \F_q$.
A function $f: \F_q^m \to \F_q$ is in $\mathcal C$ if for every line $L \subseteq \F_q^m$,
the restriction of $f$ to $L$ is a univariate polynomial of degree at most $d$.
Note that if $f$ is the evaluation table of an $m$-variate polynomial of degree $\leq d$, then
$f$ is automatically in $\mathcal C$. The surprising (and useful)
fact is that if $d$ is large and $\F_q$ has small characteristic, then
$\mathcal C$ has significantly more functions, but has the same distance as
the Reed-Muller code. This leads to its improved rate relative
to the corresponding Reed-Muller code, which only contains the evaluation tables of low degree polynomials.

Our main result is an algorithm for list-decoding and local list-decoding
of lifted codes. We show that lifted codes of distance $\delta$ can be
efficiently list-decoded and locally list-decoded (in sublinear-time)
upto their ``Johnson radius" ($1- \sqrt{1-\delta}$).
Combined with the local testability of lifted codes,
this also implies that lifted codes can be locally tested in the
high-error regime, upto the Johnson radius.

It is well known that Reed-Muller codes can be list decoded
and locally list-decoded upto the Johnson radius~\cite{PW, STV}\footnote{To locally list-decode all the way upto the Johnson bound, one actually needs a variant of~\cite{STV} given in~\cite{BK-planes}.}
\footnote{There is another regime, where $q$ is constant, in which the
Reed-Muller codes can be list-decoded beyond the Johnson bound, upto the
minimum distance. See~\cite{GKZ08, Gop10, BL14}}. Our result
shows that a lifted code, which is a natural algebraic supercode of Reed-Muller codes, despite having a vastly greater rate than the corresponding Reed-Muller code, loses absolutely nothing in terms of any (local) algorithmic decoding / testing properties.

In the appendix, we also prove two other results as part of the
basic toolkit for working with lifted codes.
\begin{itemize}
\item Explicit interpolating sets: 
For a lifted code $\mathcal C$, we give a strongly explicit subset $S$ of $\F_q^m$ such that 
for every $g: S \to \F_q$, there is a unique lifted codeword $f: \F_q^m \to \F_q$ from $\mathcal C$ with $f|_{S} = g$. The main interest in explicit interpolating sets for us is that it allows us to convert the {\em sublinear-time} local correction algorithm for lifted codes into a {\em sublinear-time} local decoding algorithm for lifted codes (earlier the known sublinear-time local correction, only implied low-query-complexity local decoding, without any associated sublinear-time local decoding algorithm).

\item Simple local decoding upto half the minimum distance:
We note that there is a simple algorithm for local decoding of lifted codes upto half the minimum distance. This is a direct translation of the elegant weighted-lines local decoding algorithm for matching-vector codes~\cite{BET10} to the Reed-Muller code / lifted codes setting.

\end{itemize}

\subsection{Methods}

We first discuss our (global) list-decoding algorithm, which generalizes the
list-decoding algorithm for Reed-Muller codes due to Pellikaan-Wu~\cite{PW}.
The main technical lemma underlying our algorithm says that codewords of lifted
codes are low-degree when viewed as univariate polynomials. This generalizes
the classical fact due to Kasami-Lin-Peterson~\cite{KLP} underlying the
Pellikaan-Wu decoding algorithm: that multivariate polynomials are low-degree
when viewed as univariate polynomials (``Reed-Muller codes are subcodes of
Reed-Solomon codes'').

The codewords of a lifted code are in general very high degree as $m$-variate polynomials over $\F_q$. There is a description of these codes in terms of spanning monomials~\cite{GKS}, but it is not even clear from this description that lifted codes have good distance. The handle that we get on lifted codes
arises by considering the big field $\F_{q^m}$, and letting
$\phi$ be an $\F_q$-linear isomorphism between $\F_{q^m}$ and $\F_q^m$. 
Given a function $f: \F_q^m \to \F_q$, we can consider the composed function
$f \circ \phi$, and view it as a function
from $\F_{q^m} \to \F_{q}$. Our technical lemma says that this function
$f \circ \phi$ is low-degree as a univariate polynomial over $\F_{q^m}$
(irrespective of the choice of the map $\phi$).

Through this lemma, we reduce the problem of list-decoding 
lifted codes over the small field $\F_q$ to the problem of
list-decoding univariate polynomials (i.e., Reed-Solomon codes)
over the large field $\F_{q^m}$. This latter problem can be
solved using the Guruswami-Sudan algorithm~\cite{GuSu}.

Our local list-decoding algorithm uses the above list-decoding algorithm.
Following~\cite{AroraSudan, STV, BK-planes}, local list-decoding of $m$-variate
Reed-Muller codes over $\F_q$ reduces to (global) list-decoding of 
 $t$-variate Reed-Muller codes over $\F_q$ (for some $t < m$).
For the list-decoding radius to approach the Johnson radius, one needs
$t \geq 2$. This is where the above list-decoding algorithm gets used.

\paragraph{Organization of this paper}

Section~\ref{section:preliminaries} introduces notation and preliminary
definitions and facts to be used in later proofs.
Section~\ref{section:global list decoding} proves our main technical result,
that lifted RS codes over domain $\F_q^m$ are low degree when viewed as
univariate polynomials over $\F_{q^m}$, as well as the consequence for global
list decoding.
Section~\ref{section:local list decoding} presents and analyzes the local list
decoding algorithm for lifted RS codes, along with the consequence for local
testability.
Appendix~\ref{section:interpolating set} describes the explicit interpolating
sets for arbitrary lifted affine-invariant codes.
Appendix~\ref{section:unique decoding} presents and analyzes the local
correction algorithm upto half the minimum distance.

\section{Preliminaries}
\label{section:preliminaries}

\subsection{Notation}

For a positive integer $n$, we use $[n]$ to denote the set $\{1,\ldots, n\}$.
For sets $A$ and $B$, we use $\{A \to B\}$ to denote the set of functions
mapping $A$ to $B$.

For a prime power $q$, $\F_q$ is the finite field of size $q$.
We think of a code $\cC \subseteq \{\F_Q^m \to \F_q\}$ as a
family of functions $f:\F_Q^m \to \F_q$, where $\F_Q$ is an extension field of
$\F_q$, but each codeword is a vector of
evaluations $(f(x))_{x \in \F_Q^m}$ assuming some canonical ordering of elements
in $\F_Q^m$; we abuse notation and say $f \in \cC$ to mean
$(f(x))_{x \in \F_Q^m} \in \cC$.

If $f:\F_q^m \to \F_q$ and line $\ell$ is a line in $\F_q^m$, this formally means
$\ell$ is specified by some $a,b \in \F_q^m$ and the restriction of $f$ to
$\ell$, denoted by $f|_\ell$, means the function $t \mapsto f(a + bt)$.
Similarly, if $P$ is a plane, then it is specified by some $a,b,c \in \F_q^m$
and the restriction of $f$ to $P$, denoted by $f|_P$, means the function
$(t,u) \mapsto f(a + bt + cu)$.

\subsection{Interpolating sets and decoding}

\begin{definition}[Interpolating set]
A set $S \subseteq \F_Q^m$ is an \emph{interpolating set} for $\cC$ if for
every $\widehat{f}: S \to \F_q$ there exists a unique $f \in \cC$ such that
$f|_S = \widehat{f}$.
\end{definition}

Note that if $S$ is an interpolating set for $\cC$, then $|\cC| = q^{|S|}$.

\begin{definition}[Local decoding]
Let $\Sigma$ be an alphabet and let $\cC: \Sigma^k \to \Sigma^n$ be an encoding 
map. A \emph{$(\rho,l)$-local decoding algorithm for $\cC$} is a
randomized algorithm
$D: [k] \to \Sigma$ with oracle access to an input word $r \in \Sigma^n$ and
satisfies the following:
\begin{enumerate}
\item%
If there is a message $m \in \Sigma^k$ such that $\delta(\cC(m), r) \le \rho$,
then for every input $i \in [k]$, we have $\Pr[D^r(i) = m_i] \ge \frac{2}{3}$.

\item%
On every input $i \in [k]$, $D^r(i)$ always makes at most $l$ queries to $r$.

\end{enumerate}
We call $\rho$ the fraction of errors decodable, or the decoding radius, and
we call $l$ the query complexity.
\end{definition}

\begin{definition}[Local correction]
Let $\cC \subseteq \Sigma^n$ be a code. A \emph{$(\rho,l)$-local correction
algorithm for $\cC$} is a randomized algorithm
$C:[n] \to \Sigma$ with oracle access to an input word $r \in \Sigma^n$ and
satisfies the following:
\begin{enumerate}
\item%
If there is a codeword $c \in \cC$ such that $\delta(c,r) \le \rho$,
then for every input $i \in [n]$, we have $\Pr[C^r(i) = c_i] \ge \frac{2}{3}$.

\item%
On every input $i \in [n]$, $C^r(i)$ always makes at most $l$ queries to $r$.

\end{enumerate}
As before, $\rho$ is the decoding radius and $l$ is the query complexity.
\end{definition}

The definition and construction of interpolating sets is motivated by the fact
that if we have an explicit interpolating set for a code $\cC$, then we have
an explicit systematic encoding for $\cC$, which allows us to easily transform
a local correction algorithm into a local decoding algorithm.

\begin{definition}[List decoding]
Let $\cC \subseteq \Sigma^n$ be a code. A \emph{$(\rho, L)$-list decoding algorithm for $\cC$} is an algorithm which takes as input a received word
$r \in \Sigma^n$ that outputs a list $\cL \subseteq \Sigma^n$ of size
$|\cL| \le L$ containing all $c \in \cC$ such that $\delta(c,r) \le \rho$.
The parameter $\rho$ is the \emph{list-decoding radius} and $L$ is the
\emph{list size}.
\end{definition}

\begin{definition}[Local list decoding]
Let $\cC \subseteq \Sigma^n$ be a code. A \emph{$(\rho, L, l)$-local list
decoding algorithm for $\cC$} is a randomized algorithm $A$ with oracle access
to an input word $r \in \Sigma^n$ and outputs a collection of randomized
oracles $A_1,\ldots,A_L$ with oracle access to $r$ satisfying the following:
\begin{enumerate}

\item%
With high probability, it holds that for every $c \in \cC$ such that
$\delta(c,r) \le \rho$, there exists a $j \in [L]$ such that for every
$i \in [n]$, $\Pr[A_j^r(i) = c_i] \ge \frac{2}{3}$.

\item%
$A$ makes at most $l$ queries to $r$, and on any input $i \in [n]$ and for every
$j \in [L]$, $A_j^r$ makes at most $l$ queries to $r$.

\end{enumerate}
As before, $\rho$ is the \emph{list decoding radius}, $L$ is the \emph{list
size}, and $l$ is the \emph{query complexity}.
\end{definition}

\subsection{Affine-invariant codes}

\begin{definition}[Affine-invariant code]
A code $\cC \subseteq \{\F_Q^m \to \F_q\}$ is \emph{affine-invariant} if for
every $f \in \cC$ and affine permutation $A:\F_Q^m \to \F_Q^m$, the function
$x \mapsto f(A(x))$ is in $\cC$.
\end{definition}

\begin{definition}[Degree set]
For a function $f:\F_Q \to \F_q$, written as $f = \sum_{d=0}^{Q-1} f_d X^d$,
its \emph{support} is $\supp(f) := \{d \in \{0,\ldots,Q-1\} \mid f_d \ne 0\}$.
If $\cC \subseteq \{\F_Q \to \F_q\}$ is an affine-invariant code, then its
\emph{degree set} $\Deg(\cC)$ is
\[
\Deg(\cC) := \bigcup_{f \in \cC} \supp(f).
\]
\end{definition}

\begin{proposition}[\cite{BGMSS}]
\label{proposition:affine-invariant code dimension}
If $\cC \subseteq \{\F_{q^m} \to \F_q\}$ is a linear affine-invariant code, 
then $\dim_{\F_q}(\cC) = |\Deg(\cC)|$.
\end{proposition}

In particular, if $S$ is an interpolating set for an affine-invariant code
$\cC \subseteq \{\F_{q^m} \to \F_q\}$, then $|S| = |\Deg(\cC)|$.
Proposition~\ref{proposition:affine-invariant code dimension} will be used in
Appendix~\ref{section:interpolating set}.

\subsection{Lifted codes}

\begin{definition}[Lift]
Let $\cC \subseteq \{\F_q \to \F_q\}$ be an affine-invariant code.
For integer $m \ge 2$, the \emph{$m$-th dimensional lift of $\cC$},
$\Lift_m(\cC)$, is the code
\[
\Lift_m(\cC) :=
\{f:\F_q^m \to \F_q \mid
f|_{\ell} \in \cC~\text{for every line $\ell$ in $\F_q^m$}\}
\]
\end{definition}

Let $\RS(q, d)$ be the Reed-Solomon code  of degree $d$ over $\F_q$,
\[
\RS(q,d) :=
\{f:\F_q \to \F_q \mid \deg(f) \le d\}.
\]
\begin{definition}[Lifted Reed-Solomon code]
The \emph{$m$-variate lifted Reed-Solomon code of degree $d$ over $\F_q$} is the
code
\[
\LiftedRS(q, d, m) := \Lift_m(\RS(q, d)).
\]
\end{definition}

For positive integers $d, e$, we say \emph{$e$ is in the $p$-shadow of $d$}, or
$e \le_p d$, if $d$ dominates $e$ digit-wise in base $p$: in other words, if
$d = \sum_{i \ge 0} d^{(i)} p^i$ and $e = \sum_{i \ge 0} e^{(i)} p^i$ are the
$p$-ary representations, then $e^{(i)} \le d^{(i)}$ for all $i \ge 0$.
We define the notion of $p$-shadow for vectors recursively as follows.
A vector $(e_1,\ldots,e_m)$ is in the $p$-shadow of $d$, denoted by
$(e_1,\ldots,e_m) \le_p d$, if $e_1 \le_p d$ and
$(e_2,\ldots,e_m) \le_p d - e_1$. It follows easily from the definition that
if $(e_1,\ldots,e_m) \le_p d$, then $\sum_{i=1}^m e_i \le d$.
The following fact motivates these definitions.

\begin{proposition}[Lucas' theorem]
\label{proposition:lucas}
Let $e_1,\ldots,e_m$ be positive integers and $d = e_1 + \cdots + e_m$ and
let $p$ be a prime. The multinomial coefficient
${d \choose e_1,\ldots,e_m} = \frac{d!}{e_1! \cdots e_m!}$ is nonzero modulo
$p$ if only if $(e_1,\ldots,e_m) \le_p d$.
\end{proposition}

For integers $a \ge 0$ and $Q > 1$, we define the mod-star operator by
$a \modstar{Q} = 0$ if $a = 0$ and $a \modstar{Q} = b \in [Q-1]$
if $a \ne 0$ and $a \equiv b \pmod{Q-1}$.
This is motivated by the fact that $X^{d}$ defines the same function as
$X^{d \modstar{q}}$ over $\F_q$.
\begin{remark}
\label{remark:modstar}
For $b \in [Q-1]$, note that $a \modstar{Q} \le b$ if and only if
there is some integer $k \ge 0$ such that $a \in [k \cdot (Q-1) + 1,
k \cdot (Q-1) + b]$.
\end{remark}

\begin{proposition}[\cite{GKS}]
\label{proposition:lifted RS degree set}
The lifted Reed-Solomon code $\LiftedRS(q, d, m)$ is spanned by monomials
$\prod_{i=1}^m X_i^{d_i}$ such that for every $e_i \le_p d_i$, $i \in [m]$,
we have $\sum_{i=1}^m e_i \modstar{q} \le d$.
\end{proposition}

\begin{proposition}[\cite{GKS}]
\label{proposition:lifted RS distance}
The lifted Reed-Solomon code $\LiftedRS(q, d, m)$ has distance
\[
\delta(\LiftedRS(q, d, m)) \ge \delta(\RS(q, d)) - q^{-1}.
\]
\end{proposition}

\subsection{Finite field isomorphisms}
\label{subsection:isomorphism}
Let $\Tr:\F_{q^m} \to \F_q$ be the $\F_q$-linear trace map
$z \mapsto \sum_{i=0}^{m-1} z^{q^i}$.
Let $\alpha_1,\ldots,\alpha_m \in \F_{q^m}$ be linearly independent over $\F_q$
and let $\phi:\F_{q^m} \to \F_q^m$ be the map
$z \mapsto (\Tr(\alpha_1 z),\ldots,\Tr(\alpha_m z))$. Since $\Tr$ is
$\F_q$-linear, $\phi$ is an $\F_q$-linear map and in fact it is an isomorphism.
Observe that $\phi$ induces a $\F_q$-linear isomorphism
$\phi^*:\{\F_q^m \to \F_q\} \to \{\F_{q^m} \to \F_q\}$ defined by
$\phi^*(f)(x) = f(\phi(x))$ for all $x \in \F_{q^m}$.

\section{Global list decoding}
\label{section:global list decoding}
In this section, we present an efficient global list decoding algorithm for
$\LiftedRS(q, d, m)$.
Define $\alpha_1,\ldots,\alpha_m \in \F_{q^m}$, $\phi$, and $\phi^*$ as in
Section~\ref{subsection:isomorphism}.
The key new structural result, Theorem~\ref{theorem:global list decoding}, states that
$\LiftedRS(q, d, m) \subseteq \{\F_q^m \to \F_q\}$
is isomorphic to a subcode of $\RS(q^m,(d+m)q^{m-1}) \subseteq
\{\F_{q^m} \to \F_q\}$. 
In particular, this lets us list decode $\LiftedRS(q, d, m)$ by list decoding
$\RS(q^m, (d+m)q^{m-1})$ up to the Johnson radius.
We will use this algorithm for $m=2$ as a subroutine in our local list decoding
algorithm in Section~\ref{section:local list decoding}.

\subsection{Lifted Reed-Solomon codes are subcodes of Reed-Solomon codes}

We begin with a lemma on monomials in lifted Reed-Solomon codes.
We postpone the proof of this lemma to Section~\ref{subsection:lemmas}.
\begin{lemma}
\label{lemma:a<=b}
Let $h_1,\ldots,h_m$ satisfy $\prod_{i=1}^m X_i^{h_i} \in \LiftedRS(q,d,m)$,
where $d < q-m$. Write $\sum_{i=1}^m h_i = a(q-1) + b$, where $0 \le a \le m$
and $0 \le b \le d$. Then $a \le b$.
\end{lemma}

We now state and prove our main structural theorem, which shows that codewords of an $m$-variate lifted Reed-Solomon code over $\F_q$ are low degree when viewed as univariate polynomials over $\F_{q^m}$.

\begin{theorem}
\label{theorem:global list decoding}
Let $d < q - m$.
If $f \in \LiftedRS(q, d, m)$, then $\deg(\phi^*(f)) \le (d + m)q^{m-1}$.
\end{theorem}
\begin{proof}
	By Proposition~\ref{proposition:lifted RS degree set} and linearity, it suffices to prove this for a monomial $f(X_1,\ldots,X_m) =
	\prod_{i=1}^m X_i^{d_i}$, where $d_1, \ldots, d_m$ have the property that for every $e_1, \ldots, e_m$ with $e_i \leq_p d_i$, we have $\sum_{i=1}^m e_i \modstar{q}\leq d$.


For $z \in \F_{q^m}$, by the multinomial theorem we get the following expansion:
\begin{align*}
	\phi^*(f)(z) &= f(\phi(z))\\
			    &= \prod_{i=1}^m (\Tr(\alpha_i z))^{d_i}\\
	      &= \prod_{i=1}^m  (\sum_{k=0}^{m-1}(\alpha_i z)^{q^k})^{d_i}\\
	      &= \prod_{i=1}^m \left( \sum_{e_{i,0}, e_{i,1}, \ldots, e_{i,m-1} \mbox{ s.t. } \sum_j e_{i,j} = d_i} { d_i \choose e_{i,0}, \ldots, e_{i,m-1} } \cdot \prod_{j=1}^m (\alpha_i z)^{e_{i,j} q^j} \right)\\
	      &= \sum_{( e_{i,j} )_{1 \leq i \leq m, 0 \leq j \leq m-1} \mbox{ s.t. } \sum_{j} e_{i,j} = d_i} \left( \prod_{i} { d_i \choose e_{i,0}, \ldots, e_{i,m-1} }  \prod_{j} (\alpha_i z)^{e_{i,j} q^j}\right) \\
	      &= \sum_{( e_{i,j} )_{1 \leq i \leq m, 0 \leq j \leq m-1} \mbox{ s.t. } \sum_{j} e_{i,j} = d_i} \left(  \left( \prod_{i} { d_i \choose e_{i,0}, \ldots, e_{i,m-1} }  \prod_{j} (\alpha_i)^{e_{i,j} q^j}\right) \cdot z^{\sum_{j=0}^{m-1} (\sum_{i=1}^{m} e_{i,j} ) \cdot q^j} \right).
\end{align*}
We now use Lucas' theorem to understand the multinomial coefficients, (in a similar manner to Lemma B.2 and Proposition 2.8 in~\cite{GKS}), and this tells us that many terms in this sum equal $0$.
So we get that $\phi^*(f)(z)$ is of the form:
\begin{align*}
	\phi^*(f)(z) &= \sum_{(e_{ij})_{1 \leq i \leq m, 0 \leq j \leq m-1} \mbox{ s.t. } e_{ij} \leq_{p} d_i }  (\ldots) \cdot z^{\sum_{j=0}^{m-1}  (\sum_{i=1}^{m} e_{i,j} ) q^{j} }.	
\end{align*}
To conclude the proof of this theorem, we just need to show that
the only monomials $z^t$ that appear in the above expression are all
such that $t \modstar{q^m}$ is at most $(d+m) \cdot q^{m-1}$.
Concretely, we need to show that
whenever $(e_{i,j})_{1 \leq i \leq m, 0 \leq j \leq m-1}$ satisfy
(1) $e_{i,j} \leq_p d_i$ for all $i,j$, and (2) $\sum_{j=0}^{m-1} e_{i,j} = d_i$, then
we have the bound
$$ E := \sum_{j=0}^{m-1} \left(\sum_{i=1}^m e_{i,j}\right) q^{j} \modstar{q^m} \le (d+m)q^{m-1}.$$

Recall that Proposition~\ref{proposition:lifted RS degree set}
allowed us to assume that $d_1, \ldots, d_m$ have the property that
for every $e_i \le_p d_i$, $i \in [m]$, we have
$\sum_{i=1}^m e_i \modstar{q} \le d$.
Therefore, $\sum_{i=1}^m e_{i,m-1} = a(q-1) + b$ for some $0 \le a \le m$
and $0 \le b \le d$.

We now proceed to give upper and lower bounds on $E$, which will then enable us to show that $E \modstar{q^m} \leq (d+m)q^{m-1}$.
We start with the upper bound:
\begin{eqnarray*}
E &=& q^{m-1} \sum_{i=1}^m e_{i, m-1} + \sum_{j=0}^{m-2} \sum_{i=1}^m e_{i,j}
q^{j} \\
&\le& q^{m-1} \sum_{i=1}^m e_{i, m-1} + q^{m-2} \sum_{j=0}^{m-2} \sum_{i=1}^m
e_{i,j} \\
&\le& q^{m-1} \cdot (a(q-1) + d) + q^{m-2} \sum_{j=0}^{m-2}\sum_{i=1}^m q \\
&=& aq^{m-1}(q-1) + (d + m)q^{m-1} \\
&\le& a(q^m-1) + (d + m) q^{m-1}.
\end{eqnarray*}
We proceed with the lower bound. If $a = 0$, then $E \ge 0$.
Suppose $a \ge 1$. Since $\le_p$ is transitive, by
Proposition~\ref{proposition:lifted RS degree set}, the monomial $\prod_{i=1}^m
X_i^{e_{i,m-1}} \in \LiftedRS(q,d,m)$. Recall that $\sum_{i=1}^m e_{i, m-1} = a(q-1) + b$. Thus by Lemma~\ref{lemma:a<=b},
$a \le b$. Therefore,
\begin{eqnarray*}
E &=& q^{m-1} \sum_{i=1}^m e_{i, m-1} + \sum_{j=0}^{m-2} \sum_{i=1}^m e_{i,j}
q^{j} \\
&\ge& q^{m-1} \sum_{i=1}^m e_{i, m-1} \\
&=& q^{m-1} (a(q-1) + b) \\
&=& aq^m + (b-a)q^{m-1} \\
&=& a(q^m-1) + (b-a)q^{m-1} + a \\
&\ge& a(q^m-1) + 1.
\end{eqnarray*}

To summarize, if $a = 0$, then $0 \leq E \leq (d+m)\cdot q^{m-1}$,
and if $a \geq 1$, then $a(q^m-1) + 1 \le E \le a(q^m-1) + (d+m)q^{m-1}$.
In both cases, we get that
$E \modstar{q^m} \le (d+m)q^{m-1}$, as desired.
\end{proof}

\begin{corollary}
\label{corollary:global list decoding}
There is a polynomial time global list decoding algorithm for $\LiftedRS(q,d,m)$
that decodes up to $1 - \sqrt{\frac{d+m}{q}}$ fraction errors.
In particular, if $m = O(1)$ and $d = (1 - \delta)q$, then
$\delta(\LiftedRS(q,d,m)) = \delta - o(1)$
and the list decoding algorithm decodes up to $1 - \sqrt{1 - \delta} - o(1)$ 
fraction errors as $q \to \infty$.
\end{corollary}
\begin{proof}
Given $r:\F_q^m \to \F_q$, convert it to $r' = \phi^*(r) \in \F_{q^m}$, and
then run the Guruswami-Sudan list decoder for $\RS := \RS(q^m, (d+m)q^{m-1})$
on $r'$ to obtain a list $\cL$ with the guarantee that any $f \in \RS$ with
$\delta(r', f) \le 1 - \sqrt{\frac{d+m}{q}}$ lies in $\cL$.
We require that any $f \in \LiftedRS(q,d,m)$ satisfying
$\delta(r, f) \le 1 - \sqrt{\frac{d+m}{q}}$ also satisfies
$\phi^*(f) \in \cL$, and this follows
immediately from Theorem~\ref{theorem:global list decoding}.
The fact that $\delta(\LiftedRS(q, d, m)) = \delta - o(1)$ when $m = O(1)$ and
$q = (1 - \delta)q$ follows immediately from
Proposition~\ref{proposition:lifted RS distance}.
\end{proof}

\subsection{Proof of Lemma~\ref{lemma:a<=b}}
\label{subsection:lemmas}

We begin with three simple claims about the $\leq_p$ relation.

\begin{claim}
\label{claim:one}
If $e \le_p h_1 + \cdots + h_m$, then there exist $e_1,\ldots,e_m$ such that
$e_i \le_p h_i$ for each $i \in [m]$ and $e_1 + \cdots + e_m = e$.
\end{claim}
\begin{proof}
The coefficient of $X^e$ in $(1 + X)^{h_1 + \cdots + h_m}$ is
$\sum_{e_1 + \cdots + e_m = e} \prod_{i=1}^m {h_i \choose e_i}$.
By Proposition~\ref{proposition:lucas}, the hypothesis implies that this
coefficient is nonzero modulo $p$, hence there is some choice of
$e_1 + \cdots + e_m = e$ such that $\prod_{i=1}^m {h_i \choose e_i}$ is nonzero
modulo $p$. By Proposition~\ref{proposition:lucas}, $e_i \le_p h_i$ for each
$i \in [m]$.
\end{proof}

\begin{claim}
\label{claim:two}
Let $c \ge 1$ and $k \le p^c/2$. If $0 \le x \le p^c - 2k + 1$, then there
exists $0 \le i \le k-1$ such that $x+i \le_p p^c - k$.
\end{claim}
\begin{proof}
Let $n := p^c - k$. We have the identity
\[
{n+k-1 \choose x+k-1} = \sum_{i=0}^{k-1} {n \choose x+i} {k-1 \choose i}
\]
from the fact that the LHS counts the number of ways of choosing $x+k-1$
elements from $[n+k-1]$, whereas the RHS counts the same thing by picking
$x+i$ elements from $[n]$ and picking $(k-1)-i$ elements from
$\{n+1,\ldots,n+k-1\}$.
The LHS is ${p^c-1 \choose x+k-1} \not\equiv 0 \pmod p$ by
Proposition~\ref{proposition:lucas}. Using the identity above,
there must be some $i$ such that ${p^c - k \choose x+i} = {n \choose x+i}
\not\equiv 0 \pmod p$. Again, by Proposition~\ref{proposition:lucas},
$x+i \le_p p^c - k$.
\end{proof}

\begin{claim}
\label{claim:three}
If $N = a(q-1) + b$, where $1 \le b < a \le m$ and $q$ is a power of prime $p$,
then there exists $e \le_p N$ such that $b < e \le m$.
\end{claim}
\begin{proof}ñ
Write $q = p^s$ and $a = p^c - r$, where $0 \le r < p^c$.
Then $N = aq - p^c + r+b = (a-1)q + (p-1)\sum_{i=c}^{s-1} p^i + (r+b)$.
But $r+b = p^c-(a-b) < p^c$, therefore $r+b \le_p N$. Therefore, it suffices
to find $e \le_p r+b$ such that $b < e \le a \le m$.
If $r+b \le a$, then we can simply take $e := r+b$.
Otherwise, if $a < r+b$, then $a-b < p^c/2$, for if not, then
$a \ge p^c/2$ and $r+b = p^c-(a-b) \le p^c/2$ and therefore $r+b \le a$,
a contradiction. By Claim~\ref{claim:two}, there exists $i \in [a-b]$ such
that $b+i \le_p p^c-(a-b) = r+b$. Set $e := b+i$.
\end{proof}

We can now complete the proof of Lemma~\ref{lemma:a<=b}.
\begin{proof}[Proof of Lemma~\ref{lemma:a<=b}]
If $a=0$, then the result trivially holds. Suppose $a \ge 1$. Then $b \ge 1$.
Suppose, for the sake of contradiction, that $a > b$.
By Claim~\ref{claim:three}, there exists $e \le_p h_1 + \cdots + h_m$ such
that $b < e \le m$. By Claim~\ref{claim:one}, there exist $e_1,\ldots,e_m$
such that $e_i \le_p h_i$ for $i \in [m]$ and $e_1 + \cdots + e_m = e$.
For $i \in [m]$, define $b_i := h_i - e_i$. Then $b_i \le_p h_i$, and so
by Proposition~\ref{proposition:lifted RS degree set} we have
$\sum_{i=1}^m b_i \modstar{q} \le d$. On the other hand,
$\sum_{i=1}^m b_i = \sum_{i=1}^m h_i - \sum_{i=1}^m e_i = a(q-1) + b - e$.
We can lower bound this by
\[
a(q-1) + b - e \ge a(q-1) + b - m \ge (a-1)(q-1) + q - m > (a-1)(q-1) + d
\]
and upper bound this by
\[
a(q-1) + b - e \le a(q-1) - 1 < (a-1)(q-1) + (q-1)
\]
and so $\sum_{i=1}^m b_i \modstar{q} > d$, a contradiction.
\end{proof}

\section{Local list decoding}
\label{section:local list decoding}

In this section, we present a local list decoding algorithm for 
$\LiftedRS(q,d,m)$, where $d = (1 - \delta)q$ which decodes up to radius
$1 - \sqrt{1 - \delta} - \epsilon$ for any constant $\epsilon > 0$, with
list size $\poly(\frac{1}{\epsilon})$ and query complexity $q^3$.

\paragraph{Local list decoder:}
Oracle access to received word $r: \F_q^m \to \F_q$.
\begin{enumerate}
\item%
Pick a random line $\ell$ in $\F_q^m$.

\item%
Run Reed-Solomon list decoder (e.g.\ Guruswami-Sudan) on $r|_\ell$ from
$1 - \sqrt{1 - \delta} - \frac{\epsilon}{2}$ fraction errors to get list
$g_1, \ldots, g_L: \F_q \to \F_q$ of Reed-Solomon codewords.

\item%
For each $i \in [L]$, output $\textsf{Correct}(A_{\ell, g_i})$

\end{enumerate}

where $\textsf{Correct}$ is a local correction algorithm for the lifted codes
for $0.1 \delta$ fraction errors, and $A$ is an oracle which takes as advice
a line and a univariate polynomial and simulates oracle access to a function
which is supposed to be $\ll 0.1\delta$ close to a lifted RS codeword.

\paragraph{Oracle $A_{\ell, g}(x)$:}

\begin{enumerate}
\item%
If $\ell$ contains $x$, i.e.\ $\ell = \{a + bt \mid t \in \F_q\}$ for some
$a,b \in \F_q^m$ and $x = a + bt$, then output $g(t)$.

\item%
Otherwise, let $P$ be the plane containing $\ell$ and $x$,
parametrized by $\{a + bt + (x - a)u \mid t,u \in \F_q\}$.
\begin{enumerate}
\item%
Use the global list decoder for bivariate lifted RS code given above to list
decode $r|_P$ from $1 - \sqrt{1 - \delta} - \frac{\epsilon}{2}$ fraction errors
and obtain a list $\cL$.

\item%
If there exists a unique $h(t,u) \in \cL$ such that $h|_\ell = g$, output
$h(0,1)$, otherwise fail.
\end{enumerate}

\end{enumerate}

\paragraph{Analysis:}
To show that this works, we just have to show that, with high probability
over the choice of $\ell$, for every lifted RS codeword $f$ such that
$\delta(r, f) \le 1 - \sqrt{1 - \delta} - \epsilon$, there is $i \in [L]$
such that $\textsf{Correct}(A_{\ell, g_i}) = f$, i.e.\
$\delta(A_{\ell, g_i}, f) \le 0.1 \delta$.

We will proceed in two steps:
\begin{enumerate}
\item%
First, we show that with high probability over $\ell$, there is some
$i \in [L]$ such that $f|_\ell = g_i$.

\item%
Next, we show that $\delta(A_{\ell, f|_\ell}, f) \le 0.1\delta$.

\end{enumerate}

For the first step, note that $f|_\ell \in \{g_1,\ldots,g_L\}$ if
$\delta(f|_\ell, r|_\ell) \le 1 - \sqrt{1 - \delta} - \frac{\epsilon}{2}$.
Note that $\delta(f|_\ell, r|_\ell)$ has mean $1 - \sqrt{1 - \delta} - \epsilon$
with variance less than $\frac{1}{q}$ (by pairwise independence of points on
a line), so by Chebyshev's inequality the probability that
$\delta(f|_\ell, r|_\ell) \le 1 - \sqrt{1 - \delta} - \frac{\epsilon}{2}$ is
$1 - o(1)$.

For the second step, we want to show that
$\Pr_{x \in \F_q^m}[A_{\ell,f|_\ell}(x) \ne f(x)] \le 0.1 \delta$.
First consider the probability when we randomize $\ell$ as well.
We get $A_{\ell, f|_\ell}(x) = f(x)$ as long as $f|_P \in \cL$ and
no element $h \in \cL$ has $h|_\ell = f|_\ell$.
With probability $1 - o(1)$, $\ell$ does no contain $x$, and conditioned on
this, $P$ is a uniformly random plane. It samples the space $\F_q^m$ well,
so with probability $1 - o(1)$ we have
$\delta(f|_P, r|_P) \le 1 - \sqrt{1 - \delta} - \frac{\epsilon}{2}$ and
hence $f|_P \in \cL$.
For the probability that no two codewords in $\cL$ agree on $\ell$,
view this as first choosing $P$, then choosing $\ell$ within $P$.
The list size $|\cL|$ is a constant, polynomial in $1/\epsilon$.
So we just need to bound the probability that two bivariate lifted RS codewords
agree on a uniformly random line. This is the same as the probability
that a uniformly random line in $\F_q^2$ is contained in the agreement set of
two fixed bivariate lifted RS codewords, which we know has size at most $(1-\delta + o(1))q^2$. By a standard second moment bound, this probability is at most
$\Theta(1/q)$. Thus, with probability $1 - o(1)$, $f|_P$ is the unique codeword in $\cL$
which is consistent with $f|_\ell$ on $\ell$. Therefore,
\begin{eqnarray*}
\Pr_{\ell}\left[\delta(A_{\ell, f|_\ell}, f|_\ell) > 0.1\delta\right]
&=& \Pr_{\ell}\left[\Pr_{x}[A_{\ell, f|_\ell}(x) \ne f(x)] > 0.1\delta\right] \\
&\le& \frac{\Pr_{\ell, x}[A_{\ell, f|_\ell}(x) \ne f(x)]}{0.1\delta}  \\
&=& o(1).
\end{eqnarray*}

As a corollary, we get the following testing algorithm.

\begin{theorem}
For any $\alpha < \beta < 1 - \sqrt{1 - \delta}$,
there is an $O(q^4)$-query algorithm which, given oracle access to a function
$r: \F_q^m \to \F_q$, distinguishes between the cases where $r$ is
$\alpha$-close to $\LiftedRS(q,d,m)$ and where $r$ is $\beta$-far.
\end{theorem}
\begin{proof}
Let $\rho = (\alpha + \beta)/2$ and let $\epsilon = (\beta - \alpha)/8$,
so that $\alpha = \rho - 4\epsilon$ and $\beta = \rho + 4\epsilon$.
Let $T$ be a local testing algorithm for $\LiftedRS(q,d,m)$ with query
complexity $q$, which distinguishes between codewords and words that are
$\epsilon$-far from the code.
The algorithm is to run the local list decoding algorithm on $r$ with error
radius $\rho$ such that $\alpha < \rho < \beta$, to obtain a list of oracles
$M_1,\ldots,M_L$. For each $M_i$, we use random sampling to estimate the
distance between $r$ and the function computed by $M_i$ to within $\epsilon$
additive error, and keep only the ones with estimated distance less than $\rho + \epsilon$. Then, for each remaining $M_i$,
we run $T$ on $M_i$. We accept if $T$ accepts some $M_i$, otherwise we reject.

If $r$ is $\alpha$-close to $\LiftedRS(q,d,m)$, then it is $\alpha$-close to
some codeword $f$, and by the guarantee of the local list decoding algorithm
there is some $j \in [l]$ such that $M_j$ computes $f$. Moreover, this $M_j$
will not be pruned by our distance estimation. Since $f$ is a codeword,
this $M_j$ will pass the testing algorithm and so our algorithm will accept.

Now suppose $r$ is $\beta$-far from $\LiftedRS(q,d,m)$, and consider any oracle
$M_i$ output by the local list decoding algorithm and pruned by our distance
estimation. The estimated distance between $r$ and the function computed by 
$M_i$ is at most $\rho + \epsilon$, so the true distance is at most
$\rho + 2\epsilon$. Since $r$ is $\beta$-far from any codeword, that means
the function computed by $M_i$ is $(\beta - (\rho + 2\epsilon)) > \epsilon$-far
from any codeword, and hence $T$ will reject $M_i$.

All of the statements made above were deterministic, but the testing algorithm
$T$ and distance estimation are randomized procedures. However, at a price of
constant blowup in query complexity, we can make their failure probabilities
arbitrarily small constants, so that by a union bound the distance estimations
and tests run by $T$ simultaneously succeed with large constant probability.
\end{proof}

\section*{Acknowledgements}
We thank the anonymous reviewers for their helpful and insightful comments.

\bibliography{multcode-survey}
\bibliographystyle{alpha}

\appendix

\section{Interpolating set for affine-invariant codes}
\label{section:interpolating set}

In this section, we present, for any affine-invariant code
$\cC \subseteq \{\F_q^m \to \F_q\}$, an explicit interpolating set
$S_{\cC} \subseteq \F_q^m$, i.e.\ for any $\widehat{f}:S_{\cC} \to \F_q$
there exists a unique $f \in \cC$ such that $f|_{S_{\cC}} = \widehat{f}$.

Define $\alpha_1,\ldots,\alpha_m \in \F_{q^m}$, $\phi$, and $\phi^*$ as in
Section~\ref{subsection:isomorphism}.
It is straightforward
to verify that if $\cC \subseteq \{\F_q^m \to \F_q\}$ and $S \subseteq \F_{q^m}$
is an interpolating set for $\phi^*(\cC)$, then $\phi(S)$ is an interpolating
set for~$\cC$.

\begin{theorem}
\label{theorem:interpolating set}
Let $\cC \subseteq \{\F_q^m \to \F_q\}$ be a nontrivial affine-invariant code
with $\dim_{\F_q}(\cC) = D$.
Let $\omega \in \F_{q^m}$ be a generator, i.e. $\omega$ has order $q^m - 1$.
Let $S = \{\omega,\omega^2,\ldots,\omega^D\} \subseteq \F_{q^m}$. Then
$\phi(S) \subseteq \F_q^m$ is an interpolating set for $\cC$.
\end{theorem}
\begin{proof}
The map $\phi$ induces a map
$\phi^*: \{\F_q^m \to \F_q\} \to \{\F_{q^m} \to \F_q\}$ defined by
$\phi^*(f) = f \circ \phi$. It suffices to show that $S$ is an interpolating
set for $\cC' \triangleq \phi^*(\cC)$. Observe that $\cC'$ is affine-invariant
over $\F_{q^m}$, and let
$\Deg(\cC') = \{i \mid \exists f \in \cC~~i \in \supp(f)\}$.
By Proposition~\ref{proposition:affine-invariant code dimension},
$\dim_{\F_q}(\cC') = |\Deg(\cC')|$, so suppose
$\Deg(\cC') = \{i_1,\ldots,i_D\}$. Every $g \in \cC'$ is of the form
$g(z) = \sum_{j=1}^D a_{j}z^{i_j}$, where $a_j \in \F_{q^m}$.
By linearity, it suffices to show that if $g \in \cC'$ is nonzero, then
$g(z) \ne 0$ for some $z \in S$. We have
\[
\begin{bmatrix}
\omega^{i_1} & \omega^{i_2} & \cdots & \omega^{i_D} \\
\omega^{2i_1} & \omega^{2i_2} & \cdots & \omega^{2i_D} \\
\vdots & \vdots & \ddots & \vdots \\
\omega^{Di_1} & \omega^{Di_2} & \cdots & \omega^{Di_D}
\end{bmatrix}
\begin{bmatrix}
a_1 \\
a_2 \\
\vdots \\
a_D
\end{bmatrix}
=
\begin{bmatrix}
g(\omega) \\
g(\omega^2) \\
\vdots \\
g(\omega^D)
\end{bmatrix}
\]
and the leftmost matrix is invertible since it's a generalized Vandermonde
matrix. Therefore, if $g \ne 0$, then the right-hand side, which is simply
the vector of evaluations of $g$ on $S$, is nonzero.
\end{proof}

\section{Local unique decoding upto half minimum distance}
\label{section:unique decoding}

\begin{theorem}
Let $\cC \subseteq \{\F_q \to \F_q\}$ be an affine-invariant code of distance
$\delta$.
For every positive integer $m \ge 2$ and for every $\epsilon > 0$,
there exists a local correction
algorithm for $\Lift_m(\cC)$ with query complexity $O(q/\epsilon^2)$ that
corrects up to $\left(\frac{1}{2} - \epsilon\right) \delta - \frac1q$
fraction errors.
\end{theorem}
\begin{proof}
Let $\Corr_{\cC}$ be a correction algorithm for $\cC$, so that for every
$f:\F_q \to \F_q$ that is $\delta/2$-close to some $g \in \cC$,
$\Corr_{\cC}(f) = g$. The following algorithm is a local correction
algorithm that achieves the desired parameters.

\paragraph{Local correction algorithm:} Oracle access to received word
$r:\F_q^m \to \F_q$.

On input $x \in \F_q^m$:
\begin{enumerate}
\item%
Let $c = \lceil\frac{4\ln 6}{\epsilon^2}\rceil$ and pick $a_1,\ldots,a_c \in \F_q^m$
independently and uniformly at random.

\item%
For each $i \in [c]$:

\begin{enumerate}
\item%
Set $r_i(t) := r(x + a_it)$.

\item%
Compute $s_i := \Corr_{\cC}(r_i)$ and $\delta_i := \delta(r_i, s_i)$.

\item%
Assign the value $s_i(0)$ a weight
$W_i := \max\left(1 - \frac{\delta_i}{\delta/2}, 0\right)$.

\end{enumerate}

\item%
Set $W := \sum_{i=1}^c W_i$. For every $\alpha \in \F_q$, let
$w(\alpha) := \frac{1}{W}\sum_{i:s_i(0) = \alpha} W_i$.
If there is an $\alpha \in \F_q$ with $w(\alpha) > \frac{1}{2}$, output
$\alpha$, otherwise fail.
\end{enumerate}

\paragraph{Analysis:}
Fix a received word $r:\F_q^m \to \F_q$ that is $(\tau - \frac{1}{q})$-close
from a codeword
$c \in \Lift_m(\cC)$, where $\tau = \left(\frac{1}{2} - \epsilon\right)\delta$.
The query complexity follows from the fact that the algorithm queries
$O(1/\epsilon^2)$ lines, each consisting of $q$ points.
Fix an input $x \in \F_q^m$. We wish to show that, with probability at least
$2/3$, the algorithm outputs $c(x)$, i.e.\ $w(c(x)) > \frac{1}{2}$.

Consider all lines $\ell$ passing through $x$. 
For each such line $\ell$, define the following:
\begin{eqnarray*}
\tau_{\ell} &:=& \delta(r|_{\ell}, c|_{\ell}) \\
s_{\ell} &:=& \Corr_{\cC}(r|_\ell) \\
\delta_{\ell} &:=& \delta(r|_{\ell}, s_{\ell}) \\
W_{\ell} &:=& \max\left(1 - \frac{\delta_{\ell}}{\delta/2}, 0 \right) \\
X_{\ell} &=&
\begin{cases}
W_{\ell} & s_{\ell} = c|_{\ell} \\
0 & s_{\ell} \ne c|_{\ell}.
\end{cases}
\end{eqnarray*}
Let $p := \Pr_{\ell}[s_{\ell} = c|_{\ell}]$. Note that if $s_{\ell} = c|_{\ell}$,
then $\delta_{\ell} = \tau_{\ell}$, otherwise
$\delta_{\ell} \ge \delta - \tau_{\ell}$.
Hence, if $s_{\ell} = c|_{\ell}$, then
$W_{\ell} \ge 1 - \frac{\tau_{\ell}}{\delta/2}$, otherwise
$W_{\ell} \le \frac{\tau_{\ell}}{\delta/2} - 1$.

Define
\begin{eqnarray*}
\tau_{\text{good}} &\:=& \E[\tau_{\ell} \mid s_\ell = c|_{\ell}] \\
\tau_{\text{bad}} &:=& \E[\tau_{\ell} \mid s_\ell \ne c|_{\ell}] \\
W_{\text{good}} &:=& \E[W_{\ell} \mid s_\ell = c|_{\ell}] 
\ge 1 - \frac{\tau_{\text{good}}}{\delta/2}\\
W_{\text{bad}} &:=& \E[W_{\ell} \mid s_\ell \ne c|_{\ell}]
\le \frac{\tau_{\text{bad}}}{\delta/2} - 1.
\end{eqnarray*}
Observe that
\begin{eqnarray*}
\E[\tau_{\ell}] &\le& \frac{1 + (\tau - \frac1q)(q-1)}{q} \le
\tau\\
\E[X_{\ell}] &=& p \cdot W_{\text{good}} \\
\E[W_{\ell}] &=& p \cdot W_{\text{good}} + (1-p) \cdot W_{\text{bad}}.
\end{eqnarray*}
We claim that
\begin{equation}
\label{equation:good weight is large}
p \cdot W_{\text{good}} \ge (1-p) \cdot W_{\text{bad}} + 2\epsilon.
\end{equation}
To see this, we start from
\[
\left(\frac{1}{2} - \epsilon \right) \delta = \tau
\ge \E[\tau_{\ell}]
= p \cdot \tau_{\text{good}} + (1-p) \cdot \tau_{\text{bad}}.
\]
Dividing by $\delta/2$ yields
\[
1 - 2\epsilon \ge p \cdot \frac{\tau_{\text{good}}}{\delta/2}
+ (1-p) \cdot \frac{\tau_{\text{bad}}}{\delta/2}.
\]
Re-writing $1 - 2\epsilon$ on the left-hand side as $p + (1-p) - 2\epsilon$ 
and re-arranging, we get
\[
p \cdot \left(1 - \frac{\tau_{\text{good}}}{\delta/2}\right)
\ge (1-p) \cdot \left(\frac{\tau_{\text{bad}}}{\delta/2} - 1\right) + 2\epsilon.
\]
The left-hand side is bounded from above by $p \cdot W_{\text{good}}$ while the
right-hand side is bounded from below by $(1-p) \cdot W_{\text{bad}} + 2\epsilon$, hence~\eqref{equation:good weight is large} follows.

For each $i \in [c]$, let $\ell_i$ be the line $\{x + a_it \mid t \in \F_q\}$.
Note that the $X_{\ell}$ are defined such that line $i$ contributes weight
$\frac{X_{\ell_i}}{W}$ to $w(c(x))$, so it suffices to show that, with
probability at least $2/3$,
\[
\frac{\sum_{i=1}^c X_{\ell_i}}{\sum_{i=1}^c W_{\ell_i}} > \frac{1}{2}.
\]
Each $X_{\ell}, W_{\ell} \in [0,1]$, so by Hoeffding's inequality,
\begin{eqnarray*}
\Pr\left[\left|\frac{1}{c} \sum_{i=1}^c X_{\ell_i} - \E[X_{\ell}]\right|
> \epsilon/2 \right] &\le& \exp(-\epsilon^2 c/4) \le 1/6 \\
\Pr\left[\left|\frac{1}{c} \sum_{i=1}^c W_{\ell_i} - \E[W_{\ell}]\right|
> \epsilon/2 \right] &\le& \exp(-\epsilon^2 c/4) \le 1/6.
\end{eqnarray*}
Therefore, by a union bound, with probability at least $2/3$ we have,
after applying~\eqref{equation:good weight is large},
\begin{eqnarray*}
\frac{\sum_{i=1}^c X_i}{\sum_{i=1}^c W_i}
&\ge& \frac{\E[X_{\ell}] - \epsilon/2}{\E[W_{\ell}] + \epsilon/2} \\
&=& \frac{p \cdot W_{\text{good}} - \epsilon/2}
{p \cdot W_{\text{good}} + (1-p) \cdot W_{\text{bad}} + \epsilon/2} \\
&\ge& \frac{(1-p) \cdot W_{\text{bad}} + 3\epsilon/2}
{2(1-p) \cdot W_{\text{bad}} + 5\epsilon/2} \\
&>& \frac{1}{2}
\end{eqnarray*}
where the second to last inequality follows from \eqref{equation:good weight is large} and the fact that if $a < b$ and $x \le y$, then
$\frac{x + a}{x + b} \le \frac{y + a}{y + b}$ (here $a=-\epsilon/2$,
$b = (1-p) \cdot W_{\text{bad}} + \epsilon/2$,
$x = (1-p) \cdot W_{\text{bad}} + 2\epsilon$, and
$y = p \cdot W_{\text{good}}$).
\end{proof}

\end{document}